\documentclass[11pt]{revtex4}
\usepackage{amsmath, amsthm, amssymb}
\usepackage{color,calc}
\usepackage{hyperref}

\newcommand{\sech}{\operatorname{sech}}
 
\newcommand{\pa}{\partial}
\newcommand{\Res}{\mathrm{Res}} 
\newcommand{\Wr}{\mathrm{Wr}}
\newcommand{\af}{\alpha}
\newcommand{\dt}{\delta}
\newcommand{\ta}{\tau}
\newcommand{\tht}{\theta}

\newcommand{\Om}{\Omega}
\newcommand{\ld}{\lambda}
\newcommand{\mc}{\mathcal}

\newcommand{\invp}{D^{-1}}

\newtheorem{prop}{Proposition}
\newtheorem{lemm}{Lemma}
\newtheorem{thm}{Theorem}
\newtheorem{eg}{Example}
\newtheorem{rmk}{Remark}

\newtheorem{defn}{Definition}



\begin{document}

\title{A Generalized Dressing Approach for Solving the Extended KP and the
  Extended mKP Hierarchy}

\author{Xiaojun Liu} \email{tigertooth4@gmail.com} \affiliation{Department of
  Applied Mathematics, China Agricultural University, Beijing, 100083, PRC}

\author{Runliang Lin} \email{rlin@math.tsinghua.edu.cn} \affiliation{Department
  of Mathematical Sciences, Tsinghua University, Beijing, 100084, PRC}

\author{Bo Jin} \email{jblh98@163.com} \affiliation{Department
  of Mathematical Sciences, Tsinghua University, Beijing, 100084, PRC}

\author{Yunbo Zeng}\email{yzeng@math.tsinghua.edu.cn} \affiliation{Department of
  Mathematical Sciences, Tsinghua University, Beijing, 100084, PRC}

\begin{abstract}
  A combination of dressing method and variation of constants as well as a
  formula for constructing the eigenfunction is used to solve the extended KP
  hierarchy, which is a hierarchy with one more series of time-flow and based on
  the symmetry constaint of KP hierarchy. Similarly, extended mKP hierarchy is
  formulated and its zero curvature form, Lax representation and reductions are
  presented. Via gauge transformation, it is easy to transform dressing
  solutions of extended KP hierarchy to the solutions of extended mKP
  hierarchy. Wronskian solutions of extended KP and extended mKP hierarchies are
  constructed explicitly. 
\end{abstract}

\maketitle
\section{Introduction}

KP hierarchy is of fundamental importance not only in the theory of integrable
systems but also in mathematical physics \cite{DJKM81,DJKM82,SS82,JM83}. Besides
interest of its own, many extensions and generalizations to KP hierarchy are
also of great interest. One of these extensions is the multi-component
generalization \cite{Leur98,KL03}. Another kind of generalization is the
so-called KP equation with self-consistent sources (KPSCS), which was discovered
by Mel'nikov \cite{Mel83,Mel87,Mel88} and appeared later in \cite{ZK86} as a
multi-scale expansion of matrix analogy of the KP equation.
Recently, we proposed an approach to construct an \emph{extended KP hierarchy}
(exKPH) \cite{LZL08}. Inspired by the squared eigenfunction symmetry constraint
of KP hierarchy, we introduced the new $\ta_k$-flow by ``extending'' a specific
$t_k$-flow of KP hierarchy. Then we find the exKPH consisting of $t_n$-flow of
KP hierarchy, $\ta_k$-flow and the $t_n$-evolutions of eigenfunctions and
adjoint eigenfunctions. The commutativity of $t_n$-flow and $\ta_k$-flow gives
rise to zero curvature representation for exKPH. Duing to the introduction of
$\ta_k$-flowm the exKPH contains two time series $\{t_n\}$ and $\{\ta_k\}$ and
more components by adding eigenfunctions and adjoint eigenfunctions.

The exKPH contains the first type and second type of KPSCS presented in
\cite{Mel83,Mel87,Mel88}. By $t_n$-reduction and $\ta_k$-reduction, the exKPH
reduces to the Gelfand-Dickey hierarchy with self-consistent sources and
constrained KP hierarchy, respectively. For the $\ta_0$ case, the exKPH with
$k=0$ gives rise to the system considered in \cite{Oevel93,OC98,ANP98}.

Since symmetry constraints are common objects for integrable systems, this idea
for building exKPH was applied to many other KP like 2+1 dimensional
hierarchies, such as BKPH \cite{WLZ07-2}, CKPH \cite{WLZ07}, dispersionless
KPH\cite{WLZ07-3} and even q-deformed KP hierarchy \cite{LLZ07} and
semi-discrete system like two dimensional Toda lattice hierarchy
\cite{LZL07}. The corresponding extended (2+1)-dimensional integrable systems
and their reductions give both well-known and new integrable models.



The \emph{dressing} method is an important tool for solving GD and KP hierarchy
\cite{Dickey-bk03}. However this method can not be applied directly for solving
the ``extended'' hierarchy. For such reason, modifications to the traditional
dressing method is needed. In this paper, with the combination of variation of
constants method, a generalization to the dressing method is proposed, which is
based on the dressing method for GD and KP hierarchy \cite{Dickey-bk03} and the
approach for finding Wronskian solutions to constrained KP hierarchy
\cite{OS96}. In this way, we can solve the entire hierarchy of extended KP in an
unified and simple manner. As the special cases, the \emph{both} types of KPSCS,
i.e. (\ref{eqns:KPSCS}) and (\ref{eqns:Another-KPSCS}), are solved
simultaneously.

As another part of our paper, we propose the extended mKP hierarchy
(exmKPH). Similar with the extended KP hierarchy, we introduce the $\ta_k$-flow
by the inspiration of squared eigenfunction symmetries constraints
($q_i\pa^{-1}r_i\pa$) \cite{OC98} of mKPH. A special case of exmKPH (called
non-standard exmKPH) is obtained by choosing specific $q_i$ and $r_i$. Two types
of reductions, say $t_n$ and $\ta_k$ reduction of exmKPH are discussed. With
$t_n$ reduction, we obtain 1+1 hierarchies, including mKdV equation with
self-consistent sources. With $\ta_k$ reduction, we obtain constrained mKP
hierarchies discussed in \cite{OS93, XZ05}. A gauge transformation between exKPH
and exmKPH is presented, based on the gauge transformation constructed by
\cite{OS93,XZ06}. Since solutions for exKPH is obtained, this transformation
helps us to show the explicit formulation of Wronskian solutions for the exmKPH.

Our paper will be organized as follows. In section 2, we briefly recall the
extended KP hierarchy. In section 3, we present the generalized dressing
approach. In section 4, we present the new extended mKP hierarchy and its
reductions. In section 5, we give gauge transformation between the exKPH and
exmKPH. In section 6, we construct some solutions of exKPH and exmKPH. In the
last section, we give conclusion.

\section{The extended KP hierarchy}
\label{sec: exKP}
In order to make our paper self-contained, we introduce the extended KP
hierarchy \cite{LZL08} briefly. As well known, the Lax equation of KP hierarchy
is given by
\begin{equation}
  \label{eq:Lax}
  \pa_{t_n}L=[B_n, L],\quad \text{for} \;n\ge1,
\end{equation}
where $L = \pa+u_1\pa^{-1}+u_2\pa^{-2}+\cdots$ is a pseudo-differential operator
(PDO). $B_n=L^n_{\ge 0}$ stands for the projection of PDO to its non-negative power
part. The commutativity of $\pa_{t_n}$ and $\pa_{t_k}$ flow give the
zero-curvature equations of KP hierarchy
\begin{displaymath}
  B_{n,t_k}-B_{k, t_n}+[B_n, B_k]=0, \quad n,k>0
\end{displaymath}

It is known \cite{KSS-contr91,Cheng-92} that the squared eigenfunction symmetry
constraint given by
\begin{align*}
  L^k&=B_k+\sum_{i=1}^Nq_i\pa^{-1}r_i\\
  q_{i,t_n}&=B_n(q_i)\\
  r_{i,t_n}&=-B_n^*(r_i),\quad i=1,\ldots,N
\end{align*}
is compatible with the KP hierarchy and reduces the KP hierarchy to the
$k$-constrained KP hierarchy. Here the ${}^*$ denotes the adjoint operator. In
this paper, we use $P(f)$ to denote an action of differential operator $P$ on
the function $f$, while $P f$ means the multiplication of differential operator
$P$ and zero order differential operator $f$.

Inspired by squared eigenfunction symmetry constraints, we presented the
following integrable extended KP hierarchy (exKPH) in \cite{LZL08}.
\begin{subequations}
  \label{eqns:Lax-ext}
  \begin{align}
    &\pa_{\ta_k}L=[B_k+\sum_{i=1}^Nq_i\pa^{-1}r_i,L],\quad N\ge0 
    \label{eqn:tau-flow-extKP}\\
    &\pa_{t_n}L=[B_n,L],\quad \forall n\neq k\label{eqn:t-flow-extKP}\\
    &\pa_{t_n}q_i=B_n(q_i),\label{eqn:q-extKP} \\
    &\pa_{t_n}r_i=-B_n^*(r_i),\quad i=1,\ldots,N.\label{eqn:r-extKP}
  \end{align}
\end{subequations}

By showing the commutativity of (\ref{eqn:tau-flow-extKP}) and
(\ref{eqn:t-flow-extKP}) under (\ref{eqn:q-extKP}) and (\ref{eqn:r-extKP}),  the
zero curvature equation for exKPH (\ref{eqns:Lax-ext}) is
\begin{subequations}
  \label{eqns:zc-exKP}
  \begin{align}
    &B_{n,\ta_k}-B_{k,t_n}+[B_n,B_k]-\sum_{i=1}^N[q_i\pa^{-1}r_i,B_n]_{\ge0}=0
    \label{eqn:zc-exKP-main}\\
    &q_{i,t_n}=B_n(q_i),\label{eqn:zc-exKP-q}\\
    &r_{i,t_n}=-B_n^*(r_i), \quad i=1,\ldots,N.\label{eqn:zc-exKP-r}
  \end{align}
\end{subequations}
It is easy to see when $N=0$ equations (\ref{eqns:Lax-ext}) and
(\ref{eqns:zc-exKP}) returns back to the ordinary KP hierarchy. Under
\eqref{eqn:zc-exKP-q} \eqref{eqn:zc-exKP-r} the Lax pair for
(\ref{eqn:zc-exKP-main}) is
\begin{align}
  \label{eqn:zc-exKP-LaxPair}
  \Psi_{t_n}=B_n(\Psi),\quad \Psi_{\ta_k}=(B_k+\sum_{i=1}^Nq_i\pa^{-1}r_i)(\Psi).
\end{align}
The extended KP is closely related to the KP equation with self-consistent
sources (KPSCS), for instance:
\begin{eg}
  When $n=2$ and $k=3$, (\ref{eqns:zc-exKP}) provides the first type of KPSCS.
   \begin{subequations}
   \label{eqns:KPSCS}
    \begin{align}
      &(4u_{t}-12uu_x-u_{xxx})_x-3u_{yy}+4\sum_{i=1}^N(q_ir_i)_{xx}=0,\label{eqn:KPSCS-main}\\
      &q_{i,y}=q_{i,xx}+2uq_i,\quad i=1,\ldots,N,\\
      &r_{i,y}=-r_{i,xx}-2ur_i.
    \end{align}
  \end{subequations}
  While for $n=3$ and $k=2$, it gives the second type of  KPSCS.
  \begin{subequations}
    \label{eqns:Another-KPSCS}
    \begin{align}
      &4u_t-12uu_{x}-u_{xxx}-3\invp
      u_{yy}=3\sum_{i=1}^N[q_{i,xx}r_i-q_ir_{i,xx}+(q_ir_i)_y],
      \label{eqn:anoth-KPSCS-main}\\
      &q_{i,t}=q_{i,xxx}+3uq_{i,x}+\frac32q_i\invp u_y+\frac32 q_i\sum_{j=1}^Nq_jr_j+\frac32 u_xq_i,\\
      &r_{i,t}=r_{i,xxx}+3ur_{i,x}-\frac32r_i\invp u_y-\frac32
      r_i\sum_{j=1}^Nq_jr_j+\frac32 u_xr_i,
    \end{align}
  \end{subequations}
  respectively, where $\invp$ stands for the inverse of $\frac{d}{dx}$.
\end{eg}

\begin{rmk}
  In \cite{OC98} the author considered the following flow generated by
  \begin{align*}
    L_\ta &= [q\pa^{-1} r, L]\\
    q_{t_n} &= B_n(q)\\
    r_{t_n} & = -B_n^*(r).
  \end{align*}
  They proved the $\pa_{\ta}$-flow commutes with all the $t_n$-flow of
  (\ref{eq:Lax}). So the whole system is compatible. This system can be regarded
  as a special case of (\ref{eqns:Lax-ext}), if we consider $k=0$ and $\ta_0=\ta$.
\end{rmk}

\section{Generalized Dressing Approach and Variation of Constants for exKPH}

Inspired by \cite{Dickey-bk03} and \cite{OS96}, we consider the dressing
formulation of exKPH. We assume that $L$ operator of exKPH can be written in the
dressing form: $$L=W\pa W^{-1},$$ where $ W = 1 + w_1\pa^{-1} + w_2 \pa^{-2} +
\cdots$ is called a \emph{dressing operator}. Then in terms of $W$, exKPH can be
written as
\begin{subequations}
  \label{eqns:W-evolu}
  \begin{align}
    \pa_{t_n}W&=-L^n_-W,\quad n\neq k\label{eqn:W-tn}\\
    \pa_{\ta_k}W&=-L^k_-W+\sum_{i=1}^Nq_i\pa^{-1}r_iW\label{eqn:W-tauk}
  \end{align}
\end{subequations}
where $q_i$ and $r_i$ satisfy \eqref{eqn:q-extKP} and \eqref{eqn:r-extKP}
respectively. The simplest $W$ has finite many terms, which is $W= 1+
w_1\pa^{-1} + w_2\pa^{-2} + \cdots + w_N\pa^{-N}$.  It is equivalent to assume
that $W$ is a pure differential operator of order $N$ ( Since it is equivalent
to multiplying $\pa^N$ from right side to the above expression ).

Let $h_1,\ldots,h_N$ be linearly independent functions satisfying
\begin{equation}
  \label{eqn:h}
  \pa_{t_n}h_i=\pa^n(h_i),\quad i=1,\ldots,N.
\end{equation}
It is known \cite{Dickey-bk03} that the \emph{Wronskian determinant}
\begin{displaymath}
  \Wr(h_1,\ldots,h_N)=\left|
  \begin{matrix}
    h_1 & h_2 & \cdots & h_N\\
    h_1'& h_2'& \cdots & h_N'\\
    \vdots & \vdots & \vdots & \vdots\\
    h_1^{(N-1)} & h_2^{(N-1)} & \cdots & h_N^{(N-1)}
  \end{matrix}
  \right|
\end{displaymath}
is a $\ta$-function of KP hierarchy and the $N$-th order differential operator
given by
\begin{equation}
  \label{eqn:W}
  W=\frac1{\Wr(h_1,\cdots,h_N)}
  \left|  
    \begin{matrix}
      h_1 & h_2 & \cdots & h_N & 1\\
      h_1'& h_2'& \cdots & h_N'& \pa\\
      \vdots & \vdots & \vdots & \vdots & \vdots\\
      h_1^{(N)} & h_2^{(N)} & \cdots & h_N^{(N)}& \pa^N
    \end{matrix}
  \right|
\end{equation}
provides a dressing operator for KP hierarchy satisfying (\ref{eqn:W-tn}). The
numerator of (\ref{eqn:W}) is a formal determinant, which is denoted by
$\Wr(h_1,\cdots,h_N,\pa)$. It is understood as an expansion with respect to its
last column, in which all sub-determinants are collected on the left of the
differential symbols.

Unfortunately, the dressing operator $W$ defined by (\ref{eqn:h}) and
(\ref{eqn:W}) neither satisfy (\ref{eqn:W-tauk}), nor provide formula for $q_i$
and $r_i$. Since (\ref{eqn:W-tauk}) can be regarded as $\pa_{\ta_k}W=-L^k_-W$
with non-homogeneous term, we can provide a new dressing operator by combining
the usual dressing approach and the method of variation of constants. We also
present formulas for $q_i$ and $r_i$, which is motivated by \cite{OS96}.

Let $f_i$, $g_i$ satisfy
\begin{subequations}
\begin{align}
  \label{eq:fg}
  &\pa_{t_n}f_i=\pa^n(f_i),\quad \pa_{\ta_k}f_i=\pa^k(f_i)\quad (i=1,\ldots,N)\\
  &\pa_{t_n}g_i=\pa^n(g_i),\quad \pa_{\ta_k}g_i=\pa^k(g_i).
\end{align}
\end{subequations}
Now let $h_i$ be the linear combination of $f_i$ and $g_i$ as
\begin{equation}
  \label{eq:h'}
  h_i=f_i+\af_i(\ta_k)g_i\quad i=1,\ldots,N,
\end{equation}
with the coefficient $\af_i$ being a differentiable function of $\ta_k$. (Suppose
$h_1,\ldots, h_N$ are still linearly independent.) We call this the variation of
constants because if $\af_i$ is a constant independent of $\ta_k$ , the
formulation returns back to (\ref{eqn:h}) and (\ref{eqn:W}) for KP
case. 
Then clearly $W$ defined by (\ref{eqn:W}) and (\ref{eq:h'}) still satisfy
(\ref{eqn:W-tn}). To claim that $W$ satisfies (\ref{eqn:W-tauk}), we present
\begin{equation}\label{eqn: q-r}
  q_i=-\dot{\af}_iW(g_i)\quad
  r_i=(-1)^{N-i}\frac{\Wr(h_1,\cdots,\hat{h}_i,\cdots,h_N)}
  {\Wr(h_1,\cdots,h_N)}, i=1,\ldots, N
\end{equation}
where the hat $\hat{\;}$ means rule out this term from the Wronskian
determinant, $\dot{\af}_i=\frac{d\af_i}{d\ta_k}$, and here is a proposition.

\begin{prop}
  \label{prop:1}
  Let $W$ defined by (\ref{eqn:W}) and (\ref{eq:h'}), $L=W\pa W^{-1}$, 
  then $W$, $q_i$, $r_i$ satisfy \eqref{eqns:W-evolu}.
\end{prop}

\begin{rmk}
  Note that $r_1,\ldots,r_N$ defined in \eqref{eqn: q-r} satisfy the following
  linear equation
  \begin{equation}
    \label{eqn:cond-r}
    \sum_{i=1}^N h_i^{(j)} r_i=\dt_{j,N-1}, \quad j=0,1,\cdots, N-1,
  \end{equation}
  where $\dt_{j,N-1}$ is the \emph{Kronecker's delta} symbol.
\end{rmk}

To prove Proposition \ref{prop:1}, we need several lemmas. The first one is
given by Oevel and Strampp \cite{OS96}:
\begin{lemm}(Oevel \& Strampp)
  \label{lm:OS}
  $W^{-1}=\sum_{i=1}^N h_i\pa^{-1} r_i$.
\end{lemm}
\begin{lemm}\label{lm:simple}
  For a pure differential operator $P$ and a function $f$,
  $$\Res_{\pa}\pa^{-1}fP=P^*(f).$$
\end{lemm}

\begin{lemm}
  \label{lm:vanish}
  $W^*$ annihilates each $r_i$, i.e. $W^*(r_i)=0$ for $i=1,\ldots,N$.
\end{lemm}
\begin{proof}
  Expanding the identity $W^*(W^{-1})^*\pa^j=\pa^j$ by using Lemma \ref{lm:OS},
  and taking the residue, we have
  \begin{displaymath}
    0= \Res_{\pa}W^*(\sum_{i=1}^Nh_i\pa^{-1}r_i)^*\pa^j
    =-\Res_{\pa}W^*\sum_{i=1}^Nr_i\pa^{-1}h_i\pa^j
    =(-1)^{j+1}\sum_{i=1}^Nh_i^{(j)}W^*(r_i)
  \end{displaymath}
  According to the last equality, $W^*(r_i)$ vanishes.

\end{proof}
\begin{lemm} (\emph{Key lemma})
  \label{lm:key}
  The operator $\pa^{-1} r_i W$ is a pure differential operator for each
  $i$. Furthermore, for $1\le i,j\le N$
  \begin{equation}
    \label{eq:key-eq}
    (\pa^{-1} r_i W)(h_j)=\dt_{ij}.
  \end{equation}
\end{lemm}
\begin{proof}
  Because $(\pa^{-1} r_i W)_-=\pa^{-1} [W^*(r_i)]=0$ (by lemma \ref{lm:vanish}),
  $\pa^{-1}r_iW$ is a pure differential operator. We can define functions
  $c_{ij}=(\pa^{-1}r_iW)(h_j)$ and the $x$-derivative of $c_{ij}$ shows
  $\pa(c_{ij})=r_iW(h_j)=0$, which means $c_{ij}$ does not depend on $x$. Then
  \begin{displaymath}
    \sum_{i=1}^N h_i^{(k)}c_{ij}=\pa^k(\sum_i h_i
    c_{ij})=\pa^k(\sum_i(h_i\pa^{-1} r_i W)(h_j))
    =\pa^k(W^{-1}W)(h_j)=h_j^{(k)},
  \end{displaymath}
  so $c_{ij}=\dt_{ij}$.
\end{proof}
\begin{rmk}
  Lemma \ref{lm:key} is crucial to the proof of Propostion \ref{prop:1} and
  upcoming results. We believe it is important also because it shows the
  inherent connection between $W$, the function $h_i (\in ker W)$, and the
  function $r_i (\in ker W^*)$.
\end{rmk}

\begin{proof}[Proof of Proposition \ref{prop:1}]
  The proof of (\ref{eqn:W-tn}) was given in \cite{Dickey-bk03} (pp 80).  In a
  same way, taking $\pa_{\ta_k}$ to the identity $W(h_i)=0$, using \eqref{eq:fg}
  \eqref{eq:h'}, the definition of $q_i$ and the \emph{key lemma}, we have
  \begin{align*}
    0=&(\pa_{\ta_k}W)(h_i)+(W\pa^k)(h_i)+\dot\af_i W(g_i)\\
    =&(\pa_{\ta_k}W)(h_i)+(L^kW)(h_i)-\sum_{j=1}^Nq_j\dt_{ji}\\
    =&(\pa_{\ta_k}W+L^k_-W-\sum_{j=1}^Nq_j\pa^{-1}r_jW)(h_i).
  \end{align*}
  Since the pure differential operator acting on $h_i$ in the last expression
  has degree $<N$, it can not annihilate $N$ independent functions unless the
  operator itself vanishes. Hence (\ref{eqn:W-tauk}) is proved.
\end{proof}

\begin{thm}
  \label{thm:exKP} 
  $W$ defined by (\ref{eqn:W}) and (\ref{eq:h'}), $q_i$ and $r_i$ defined by
  (\ref{eqn: q-r}), $L=W\pa W^{-1}$ satisfy the extended KP hierarchy
  (\ref{eqns:Lax-ext}).
\end{thm}


\begin{proof}
  The proof (\ref{eqn:t-flow-extKP}) and (\ref{eqn:q-extKP}) can be found in
  \cite{Dickey-bk03}. The proof of  (\ref{eqn:tau-flow-extKP}) is
  straightforward by using proposition \ref{prop:1} 
  \begin{align*}
    L_{\ta_k}=&W_{\ta_k}\pa W^{-1}-W\pa W^{-1} W_{\ta_k} W^{-1}\\
    =&(-L^k_-+\sum_i q_i\pa^{-1}r_i)L+ L(L^k_--\sum_i q_i\pa^{-1}r_i)\\
    =&[-L^k_-+\sum_i q_i\pa^{-1}r_i,L]=[B_k+\sum_{i=1}^N q_i\pa^{-1}r_i,L].
  \end{align*}
  So it remains to prove (\ref{eqn:r-extKP}). Firstly, we see that
  $$(W^{-1})_{t_n}=-W^{-1}W_{t_n}W^{-1}=-W^{-1}(L^n-B_n)=\pa^nW^{-1}-W^{-1}B_n.$$
  Then we make substitutions to this equality by applying $W^{-1}=\sum
  h_i\pa^{-1}r_i$ at both ends, we have
  \begin{align*}
    &(W^{-1})_{t_n}=\sum h_i^{(n)}\pa^{-1}r_i+\sum h_i\pa^{-1} r_{i,t_n}\\
    &=\pa^nW^{-1}-W^{-1}B_n=\sum h_i^{(n)}\pa^{-1}r_i-\sum h_i\pa^{-1}B_n^*(r_i)
  \end{align*}
  Then $\sum h_i\pa^{-1}r_{i,t_n}=-\sum h_i\pa^{-1}B_n^*(r_i)$ implies
  \eqref{eqn:r-extKP}.
\end{proof}

\section{The Extended mKP hierarchy}

In the same way as in section \ref{sec: exKP}, the extended mKP hierarchy is
formulated. The $L$ operator of mKP hierarchy is defined by
\begin{displaymath}
  L= \pa+v_0+v_1\pa^{-1}+v_2\pa^{-2}+\cdots
\end{displaymath}
The Lax equations of mKP hierarchy is given by
\begin{equation}
  L_{t_n}=[B_n,L]\quad n\ge1,\quad B_n=L^n_{\ge1}.
\end{equation}
The commutativity of $\pa_{t_n}$ and $\pa_{t_k}$ flows gives the zero curvature
equation
\begin{displaymath}
  B_{n,t_k}-B_{k,t_n}+[B_n,B_k]=0.
\end{displaymath}

When $n=2$ and $k=3$, we get \emph{mKP equation}:
\begin{displaymath}
  4v_t-v_{xxx}+6v^2v_x-3(D^{-1}v_{yy})-6v_x(D^{-1}v_y)=0,
\end{displaymath}
where $t:=t_3$, $y:=t_2$, $v:=v_0$. Dropping $y$ dependency, we obtain the mKdV
equation.


Since the squared eigenfunction symmetry constraint given by
\begin{align*}
  L^k&=B_k+\sum q_i\pa^{-1}r_i\pa\\
  q_{i,t_n}&=B_n(q_i)\quad i=1,\cdots,N\\
  r_{i,t_n}&=-(\pa B_n \pa^{-1})^*(r_i)
\end{align*}
is consistent with the mKP hierarchy and reduces mKP hierarchy to $k$-constraint
mKP hierarchy (see \cite{OC98,OS93} and references therein), we can define
extended mKP hierarchy by introduce a new time evolution of $L$.
\begin{align*}
  L_{\ta_k}&=[B_k+\sum q_i\pa^{-1}r_i\pa,L]\\
  q_{i,t_n}&=B_n(q_i)\quad i=1,\cdots,N\\
  r_{i,t_n}&=-(\pa B_n \pa^{-1})^*(r_i)
\end{align*}

\begin{defn}
  The \emph{extended mKP hierarchy} (exmKPH) is defined by
  \begin{subequations}
    \label{eqns:exmKP-Lax}
    \begin{align}
      L_{\ta_k}&=[B_k+\sum_{i=1}^Nq_i\pa^{-1}r_i\pa, L],  \label{eqn:tau-flow-exmKP}\\
      L_{t_n}&=[B_n,L], \quad n\neq k,\label{eqn:t-flow-exmKP}\\
      q_{i,t_n}&=B_n(q_i),\quad i=1,\cdots,N \label{eqn:q-exmKP}\\
      r_{i,t_n}&=-(\pa B_n\pa^{-1})^*(r_i). 
    \end{align}
  \end{subequations}
\end{defn}
We can verify the compatibility of (\ref{eqn:tau-flow-exmKP}) and
(\ref{eqn:t-flow-exmKP}) in the same way as in \cite{LZL08}.

By using the identity
$(q_i\pa^{-1}r_i\pa)_{t_n}=[B_n,q_i\pa^{-1}r_i\pa]_{\le0}$, the zero curvature
equation for exmKPH can be written as:
\begin{subequations}
  \label{eqns:zc-exmKP}
  \begin{align}
    &B_{k,t_n}-B_{n,\ta_k}+[B_k,B_n]+\sum_{i=1}^N[q_i\pa^{-1}r_i\pa,B_n]_{\ge1}=0
    \label{eqn:zc-exmKP-main}\\
    &q_{i,t_n}=B_n(q_i),\quad i=1,\cdots, N\label{eqn:zc-exmKP-q}\\
    &r_{i,t_n}=-(\pa^{-1}B_n^*\pa)(r_i).\label{eqn:zc-exmKP-r}
  \end{align}
\end{subequations}
Under the condition (\ref{eqn:zc-exmKP-q}) and (\ref{eqn:zc-exmKP-r}), the Lax
pair for (\ref{eqns:zc-exmKP}) is given by
\begin{align*}
  \Psi_{t_n}=B_n(\Psi),\quad \Psi_{\ta_k}=(B_k+\sum_{i=1}^Nq_i\pa^{-1}r_i\pa)(\Psi).
\end{align*}

\begin{eg}
  When $n=2$ and $k=3$, we get the first type of mKP with self-consistent
  sources
  \begin{align*}
    &4v_t-v_{xxx}+6v^2v_x-3\invp v_{yy}-6v_x\invp
    v_y+4\sum_{i=1}^N(q_ir_i)_x=0,\\
    &q_{i,y}=q_{i,xx}+2vq_{i,x},\\
    &r_{i,y}=-r_{i,xx}+2vr_{i,x},
  \end{align*}
  where $t:=\ta_3$, $y:=t_2$, $v:=v_0$. When $n=3$ and $k=2$, we get the second
  type of mKPSCS
  \begin{subequations}
    \label{eqns:2nd-exmKP}
    \begin{align}
      &4v_t-v_{xxx}+6v^2v_x-3\invp v_{yy}-6v_x\invp
      v_y\notag\\
      &\quad\quad\quad\quad\quad\quad\quad\quad
      +\sum_{i=1}^N[3(q_ir_{i,xx}-q_{i,xx}r_i)-3(q_ir_i)_y-6(vq_ir_i)_x]=0,\\
      &q_{i,t}=q_{i,xxx}+3vq_{i,xx}+\frac32(\invp
      v_y)q_{i,x}+\frac32v_xq_{i,x}+\frac32v^2q_{i,x}+\frac32\sum_{j=1}^N(q_jr_j)q_{i,x},\\
      &r_{i,t}=r_{i,xxx}-3vr_{i,xx}+\frac32(\invp
      v_y)r_{i,x}-\frac32v_xr_{i,x}+\frac32v^2r_{i,x}+\frac32\sum_{j=1}^N(q_jr_j)r_{i,x},
    \end{align}
  \end{subequations}
  where $y:=\ta_2$, $t:=t_3$, $v:=v_0$.
\end{eg}

\begin{prop}
  It is easy to see that any constant is an (adj-)eigenfunction for
  (\ref{eqn:zc-exmKP-q}) and (\ref{eqn:zc-exmKP-r}). Let $N=2$, $q_1:=q$,
  $r_1=1$, $q_2=-1$ and $r_2=r$, then we have a special case of exmKPH, called
  the \emph{non-standard exmKPH}:
  \begin{subequations}
    \label{eqns:non-std-exmKP}
    \begin{align}
      &B_{k,t_n}-B_{n,\ta_k}+[B_k,B_n]+[q-r+\pa^{-1}r_x,B_n]_{\ge1}=0\\
      &q_{t_n}=B_n(q),\\
      &r_{t_n}=-(\pa^{-1}B_n^*\pa)(r).
    \end{align}
  \end{subequations}
\end{prop}

\begin{eg}
  For $n=2$ and $k=3$, let $\ta_3=t$, $t_2=y$, $v_0=v$, we get the first
  non-standard extended mKP equation
  \begin{align*}
    &4v_t-v_{xxx}+6v^2v_x-3\invp v_{yy}-6v_x\invp
    v_y+4(q-r)_x=0,\\
    &q_{y}=q_{xx}+2vq_{x},\\
    &r_{y}=-r_{xx}+2vr_{x}.
  \end{align*}
\end{eg}

\begin{eg}
  For $n=3$ and $k=2$, let $\ta_2=y$, $t_3=t$, $v_0=v$, we get the second
  non-standard extended mKP equation
  \begin{align*}
    &4v_t-v_{xxx}+6v^2v_x-3\invp v_{yy}-6v_x\invp
    v_y\notag\\
    &\qquad\qquad\qquad-3(q-r)_y-3(q-r)_{xx}-6(vq-vr)_x-6r=0,\\
   &q_t=q_{xxx}+3vq_{xx}+\frac32(\invp v_y)q_x+\frac32 v_xq_x+\frac32 v^2q_x+\frac32(q-r)q_x,\\
    &r_t=r_{xxx}-3vr_{xx}+\frac32(\invp v_y)r_x-\frac32v_xr_{x}+\frac32v^2r_x+\frac32(q-r)r_x.
  \end{align*}
\end{eg}

For the extended mKP hierarchies, there are two time series $t_n$ and $\ta_k$.
So we can consider the reduction with respect to these time series, namely, the
$t_n$-reduction and $\ta_k$-reduction for the exmKP hierarchies, see appendix
\ref{sec:red-exmKP}.

\section{Gauge Transformation between exKPH and exmKPH}

In \cite{OS93,XZ05}, authors give a gauge transformation between KP and mKP
hierarchy with self-consistent sources. This transformation can be put
parallelly to the exKPH and exmKPH case. In this section, we use $L$ and $q_i$,
$r_i$ for exKP hierarchy, and $\tilde{L}$ and $\tilde{q}_i$, $\tilde{r}_i$ for
exmKP hierarchy. The main result can be established as following
\begin{thm}
  Suppose $L$, $q_i$, $r_i$ satisfy (\ref{eqns:Lax-ext}), $f$ is a particular
  eigenfunction for Lax pair (\ref{eqn:zc-exKP-LaxPair}), then
  \begin{displaymath}
    \tilde{L}:=f^{-1} L f,\quad \tilde{q}_i:=f^{-1}q_i,\quad \tilde{r}_i:=-D^{-1}(fr_i)
  \end{displaymath}
  satisfies the exmKP hierarchy (\ref{eqns:exmKP-Lax}).
\end{thm}
We omit the proof since it is done by straightforward calculation. 

Then we
choose $$f=W(1)=(-1)^N\frac{\Wr(h_1',\cdots,h_N')}{\Wr(h_1,\cdots,h_N)}$$ as the
particular eigenfunction for Lax pair (\ref{eqn:zc-exKP-LaxPair}), where $W$ is
the dressing operator defined by (\ref{eqn:W}), and $1$ is clearly the
\emph{seed solution} for (\ref{eqn:zc-exKP-LaxPair}) when $L=\pa$,
$q_i=r_i=0$. Then starting from the $L$ $q_i$ and $r_i$ given by theorem
\ref{prop:1}, the Wronskian solution for exmKP hierarchy is
\begin{subequations}
  \label{eqn:wronskian-exmKP}
  \begin{align}
    \tilde{L} &= \frac{\Wr(h_1,\cdots,h_N,\pa)}{\Wr(h_1',\cdots,h_N')}\pa
    \left[\frac{\Wr(h_1,\cdots,h_N,\pa)}{\Wr(h_1',\cdots,h_N')}\right]^{-1}\\
    \tilde{q}_i &=
    -\dot{\af}_i\frac{\Wr(h_1,\cdots,h_N,g_i)}{\Wr(h_1',\cdots,h_N')}\\
    \tilde{r}_i
    &=\frac{\Wr(h_1',\cdots,\widehat{h_i'},\cdots,h_N')}{\Wr(h_1,\cdots,h_N)},
    \label{eqn:wronskian-exmKP-r}
  \end{align}
\end{subequations}
The equation (\ref{eqn:wronskian-exmKP-r}) can be proved by the trick of Laplace
expansion formula.

\section{Solutions for the exKP and exmKP hierarchies}
Dressing approach with variation of constants and gauge transformation give us a
simple way to construct explicit solutions to exKP and exmKP hierarchies. We
will use the second type of KPSCS and the second type of mKPSCS as the
example. The solution of the second type of KPSCS is constructed by source
generating method \cite{Wang07}, the solutions of the second type of mKPSCS is
not known yet.

\begin{eg}[Solutions for the second type of KPSCS (\ref{eqns:Another-KPSCS})]
  Let $k=2$, $n=3$ in (\ref{eqns:Lax-ext}) or (\ref{eqns:zc-exKP}), $y:=\ta_2$,
  $t:=t_3$, then (\ref{eqns:Lax-ext}) leads to second type of KPSCS
  (\ref{eqns:Another-KPSCS}). Then (\ref{eq:fg}) has following solutions
  \begin{displaymath}
    f_i=\exp(\ld_ix+\ld_i^2y+\ld_i^3t):=e^{\xi_i},\quad 
    g_i=\exp(\mu_ix+\mu_i^2y+\mu_i^3t):=e^{\eta_i}
  \end{displaymath}
  where $\ld_i$ and $\mu_i$ (for $i=1,\ldots,N$) are different parameters. We
  use a linear combinations of $f_i$ and $g_i$ with coefficients $\af_i$
  explicitly depending on $y$
  \begin{displaymath}
    h_i=f_i+\af_i(y)g_i=2\sqrt{\af_i}e^{\frac{\xi_i+\eta_i}{2}}\cosh(\Om_i)
  \end{displaymath}
  where $\Om_i=\frac{\xi_i-\eta_i}{2}-\frac12\ln(\af_i)$.
  
  \begin{enumerate}
  \item One soliton solution\\
    Let $N=1$, then we have
    \begin{displaymath}
      L=W\pa W^{-1}=\pa +
      \frac{(\ld_1-\mu_1)^2}{4}\sech^2(\Om_1)\pa^{-1}+\cdots,\quad W=\pa-h_1'/h_1
    \end{displaymath}
    The one soliton solution of (\ref{eqns:Another-KPSCS}) with $N=1$
    is 
    \begin{align*}
      u&=\frac{(\ld_1-\mu_1)^2}{4}\sech^2(\Om),\\
      q_1&=\sqrt{\af_1}_y(\ld_1-\mu_1)e^{\frac{\xi_1+\eta_1}{2}}\sech(\Om_1),\\
      r_1&=\frac{1}{2\sqrt{\af_1}}e^{-\frac{\xi_1+\eta_1}{2}}\sech(\Om_1).
    \end{align*}

  \item Two soliton solution\\
    Let $N=2$, then we have 
    \begin{displaymath}
      L=\pa+ \pa_x^2\ln\ta\pa^{-1}+\cdots
    \end{displaymath}
    Two soliton solution is $$u=\pa^2\ln\ta$$
    \begin{displaymath}
      q_1=\af_{1,y}\frac{(\ld_1-\mu_1)(\ld_2-\mu_1)}{\ta} \left(1+\af_2
      \frac{(\ld_1-\mu_2) (\mu_2-\mu_1)}{(\ld_1-\ld_2)(\ld_2-\mu_1)}e^{\chi_2}\right)e^{\eta_1}
    \end{displaymath}
    \begin{displaymath}
      q_2=\af_{2,y}\frac{(\ld_2-\mu_2)(\ld_1-\mu_2)}{\ta} \left(1+\af_1
      \frac{(\ld_2-\mu_1) (\mu_1-\mu_2)}{(\ld_2-\ld_1)(\ld_1-\mu_2)}e^{\chi_1}\right)e^{\eta_2}
    \end{displaymath}
    \begin{displaymath}
      r_1=\frac{1+\af_2e^{\chi_2}}{(\ld_1-\ld_2)\ta}e^{-\xi_1},\quad
      r_2=\frac{1+\af_1e^{\chi_1}}{(\ld_2-\ld_1)\ta}e^{-\xi_2}
    \end{displaymath}
    where $\chi_i=\eta_i-\xi_i$ ($i=1,2$) and
    \begin{displaymath}
      \ta=1+\af_1\frac{\ld_2-\mu_1}{\ld_2-\ld_1}e^{\chi_1}
      +\af_2\frac{\mu_2-\ld_1}{\ld_2-\ld_1}e^{\chi_2}
      +\af_1\af_2\frac{\mu_2-\mu_1}{\ld_2-\ld_1}e^{\chi_1+\chi_2}.
    \end{displaymath}
  \end{enumerate}
\end{eg}
\begin{eg}[Solutions for the second type of mKPSCS (\ref{eqns:2nd-exmKP})]
  Since $\tilde{L}=f^{-1} L f$. Then it is easy to see
  \begin{displaymath}
    v_0=\pa_x\ln f,
  \end{displaymath}
  where $f=W(1)=(-1)^N\frac{\Wr(h_1',\cdots,h_N')}{\Wr(h_1,\cdots,h_N)}$.
  \begin{enumerate}
  \item For $N=1$ we have 
    
    \begin{displaymath}
      \tilde{L}=\pa+\frac{\ld_1-\mu_1}2\left[\tanh(\Om_1+\tht_1)-\tanh\Om_1\right] +\cdots
    \end{displaymath}
    where $\tht_1=\ln\sqrt{\frac{\ld_1}{\mu_1}}$.
    One soliton solution for (\ref{eqns:2nd-exmKP}) with $N=1$ is
    \begin{align*}
      v&=\frac{\ld_1-\mu_1}{2}[\tanh(\Om_1+\tht_1)-\tanh(\Om_1)],\\
      \tilde{q}_1&=\pa_y(\sqrt{\af_1/(\ld_1\mu_1)})(\mu_1-\ld_1)e^{\frac{\xi_1+\eta_1}{2}}\sech(\Om_1+\tht_1),\\
      \tilde{r}_1&=-\frac1{2\sqrt{\af_1}}e^{-\frac{\xi_1+\eta_1}{2}}\sech\Om_1.
    \end{align*}

  \item For $N=2$, we have

    \begin{displaymath}
      \tilde{L}=\pa+\pa_x\ln f +\cdots
    \end{displaymath}
    where $f=\ld_1\ld_2\tilde{\ta}/\ta$, $\ta$ is defined in the previous example,
    \begin{displaymath}
      \tilde{\ta}=1+\frac{\mu_1}{\ld_1}\af_1\frac{\ld_2-\mu_1}{\ld_2-\ld_1}e^{\chi_1}
      +\frac{\mu_2}{\ld_2}\af_2\frac{\mu_2-\ld_1}{\ld_2-\ld_1}e^{\chi_2}
      +\frac{\mu_1\mu_2}{\ld_1\ld_2}\af_1\af_2\frac{\mu_2-\mu_1}{\ld_2-\ld_1}e^{\chi_1+\chi_2}.
    \end{displaymath}
      
    The two soliton solution for (\ref{eqns:2nd-exmKP}) with $N=2$ is
    \begin{align*}
      v&=\pa \ln f,\\
      \tilde{q}_1&=\af_{1,y}\frac{(\ld_1-\mu_1)(\ld_2-\mu_1)}{\ld_1\ld_2\tilde{\ta}} \left(1+\af_2
        \frac{(\ld_1-\mu_2) (\mu_2-\mu_1)}{(\ld_1-\ld_2)(\ld_2-\mu_1)}e^{\chi_2}\right)e^{\eta_1},\\
      \tilde{q}_2&=\af_{2,y}\frac{(\ld_2-\mu_2)(\ld_1-\mu_2)}{\ld_1\ld_2\tilde{\ta}} \left(1+\af_1
        \frac{(\ld_2-\mu_1) (\mu_1-\mu_2)}{(\ld_2-\ld_1)(\ld_1-\mu_2)}e^{\chi_1}\right)e^{\eta_2},\\
      \tilde{r}_1&=\frac{\ld_2+\af_2\mu_2e^{\chi_2}}{(\ld_2-\ld_1)\ta}e^{-\xi_1},\quad
      \tilde{r}_2=\frac{\ld_1+\af_1\mu_1e^{\chi_1}}{(\ld_2-\ld_1)\ta}e^{-\xi_2}.
    \end{align*}
  \end{enumerate}

\end{eg}

\section{Conclusion}
In this paper, we generalize the dressing approach for KP hierarchy to exKPH by
combining dressing approach and variation of constants, we also present formulas
for $q_i$ and $r_i$.  The traditional dressing method can not provide the
evolution of extended $\ta_k$ variable. By introducing some varying constants,
say $\af_i(\ta_k)$, we can obtain the evolution of $\ta_k$ correctly. In this
way, we constructed Wronskian solutions to the whole exKP hierarchy.


In the same way as extended KP hierarchy, we considered the extended mKP
hierarchy by introducing two series of time variables, say $t_n$ and
$\ta_k$. The first and second types of mKPSCS are obtained as the first two
non-trivial equations of exmKPH. We also considered a special case of exmKPH,
called the non-standard exmKP hierarchy, by choosing special $q_i$ and $r_i$ in
exmKPH. Two types of reductions for exmKPH and non-standard exmKPH, by dropping
down $t_n$ or $\ta_k$ dependencies respectively are discussed. Some 1+1
dimensional systems, such as mKdV with self-consistent sources and
$k$-constrained mKP hierarchy, are obtained as reductions of exmKPH.

A gauge transformation between exKPH and exmKPH is given. It helps us to obtain
explicit solutions to the exmKP hierarchy. By using this transformation, we find
soliton solutions for the exmKPH, especially for the second type of mKPSCS,
which is unknown before.

It is known that besides KP and mKP, there is another 2+1 dimensional hierarchy
associated with the PDO $L:= w\pa + w_0 + w_1\pa^{-1} + \cdots$, which has the Lax
equation
\begin{displaymath}
  L_{t_n}=[L^n_{\ge2},L].
\end{displaymath}
It is called \emph{Dym} hierarchy. Note that the symmetry constraint for this
hierarchy is known and a reciprocal transformation between mKP hierarchy and Dym
hierarchy is discussed \cite{OC98}, it is not difficult to find the extended Dym
hierarchy and give its solutions.

Another interesting question is: Can we apply the dressing method with variation
of constants to other kinds of 2+1 dimensional systems? For example, the
extended BKP and extended CKP hierarchy, the discrete extended KP hierarchy, the
q-deformed extended KP hierarchy, and the KP like hierarchy defined by using
\emph{star product} instead of PDO. The last one is worth considering, because
it admits a deformation to an extended version of dispersionless KP hierarchy.
So, if we can apply dressing approach to extended version of such hierarchy, is
it possible to deform the solutions of extended version of dispersionless KP
hierarchy? If possible, it is another way to give explicit solutions for
dispersionless hierarchy, compared with implicit solution given by hydograph
method. We will consider such problems in the future.

\section*{Acknowledgement}
The author LXJ thanks Dr. Chaozhong Wu for his idea for the proof of Lemma
\ref{lm:key}. This work was supported by National Basic Research Program of
China (973 Program) (2007CB814800), National Natural Science Foundation of China
(grand No. 10601028 and grand No. 10801083).

\appendix\section{Reductions of exmKP hierarchies}
\label{sec:red-exmKP}
\subsection{The $t_n$-reduction}
The $t_n$-reduction is given by
\begin{equation}
  \label{eq:n-rd-mKP}
  L^n_{\le 0}=0, \quad \text{or} \quad L^n=B_n.
\end{equation}
Then we have
\begin{displaymath}
  (L^n)_{t_n}=[B_n, L^n]=0,\quad B_{n,t_n}=0
\end{displaymath}
$L$ is independent of $t_n$ and $B_n(q_i)=L^n(q_i)=\ld_i^n(q_i)$,
$-(\pa^{-1}B_n^*\pa)(r_i)=\ld_{i}^n(r_i)$. In such case, the submanifold
$L^n_{\le0}=0$ is invariant under $\ta_k$-flow, namely,
\begin{align*}
  &(L^n_{\le0})_{\ta_k}=[B_k+\sum_i q_i\pa^{-1}r_i\pa, L^n]_{\le0}\\
  =&[B_k, B_n]_{\le0}+[\sum_i q_i\pa^{-1}r\pa, B_n]_{\le0}\\
  =&(\sum_iq_i\pa^{-1}r_i\pa B_n)_{\le0}-\sum_i (B_n q_i\pa^{-1}r_i\pa)_{\le0}
\end{align*}
Note that for a pure differential operator $B$ without $\pa^0$ term
$$(\pa^{-1}rB)_{\le0}=\pa^{-1}(\pa^{-1}B^*(r))\pa.$$ So the last express
equals to zero:
\begin{align*}
  &-\sum_iq_i\pa^{-1}(\pa^{-1}B_n^*\pa(r_i))\pa - \sum_i B_n(q_i)\pa^{-1}r_i\pa\\
  &=\sum_i \ld_i^nq_i\pa^{-1} r_i\pa-\sum_i \ld_i^nq_i\pa^{-1}r_i\pa=0.
\end{align*}
Therefore the exmKP hierarchy (\ref{eqns:exmKP-Lax}) can be reduced to this
invariant submanifold, the $t_n$ dependency is dropped. We call the following
reduction the \emph{$t_n$-reduction} for the exmKP hierarchy,
\begin{subequations}
  \label{eqns:nRd-stdmKP}
  \begin{align}
    \mc{L}_{\ta_k}&=[\mc{L}^{k/n}_{\ge1}+\sum_{i=1}^Nq_i\pa^{-1}r_i\pa,\mc{L}]
    \label{eqn:nRd-stdmKP}\\
    \ld_i^nq_i&=\mc{L}(q_i)\label{eqn:nRd-stdmKP-egnfn}\\
    \ld_i^nr_i&=-\pa^{-1}\mc{L}^*\pa(r_i),\quad
    i=1,\ldots,N\label{eqn:nRd-stdmKP-adjEgnfn}
  \end{align}
  where $\mc{L}=L^n=\pa^n+V_{n-2}\pa^{n-1}+\cdots+V_{0}\pa$.
\end{subequations}
\begin{rmk}
  (\ref{eqns:nRd-stdmKP}) can be regarded as (1+1)-dimensional integrable
  soliton hierarchy with self-consistent sources.
\end{rmk}
\begin{eg} For $n=2$, $k=3$, $\mc{L}=\pa^2+V_0\pa$, (\ref{eqns:nRd-stdmKP})
  gives \emph{the first type of mKdV equation with self-consistent souces.}
  \begin{align*}
    &V_t=\frac14 V_{xxx}-\frac34 V^2V_x-2\sum_i(q_ir_i)_x\\
    &q_{i,xx}+Vq_{i,x}=\ld_i^2q_i\\
    &-r_{i,xx}+Vr_{i,x}=\ld_i^2r_i
  \end{align*}
  where $V:=V_0$, $t:=\ta_3$. 

  For $n=3$, $k=2$, $\mc{L}=\pa^3+V_1\pa^2+V_0\pa$, (\ref{eqns:nRd-stdmKP})
  gives
  \begin{align*}
    &V_{tt}=-\frac13 V_{xxxx}-\frac23 V_{xx}D^{-1}V_{t}-\frac23
    V_xV_t+\frac23(V^3)_{xxx}
    \\ &\quad\quad\quad
    +\sum_i \left(3(q_ir_{i,x}-q_{i,x}r_i)_{xx}-(Vq_ir_i)_{xx}-3(q_ir_i)_{xt}\right)\\
    &q_i'''+Vq_i''+Uq_i'=\ld_i^3q_i\\
    &r_i'''-Vr_i''+(V-U')r_i'=\ld_i^3r_i
  \end{align*}
  where $t=\ta_2$, $V=V_1$, $U=\frac12 D^{-1}V_t+\frac12
  V_x+\frac16V^2+\sum_i\frac32 (q_ir_i)$.
\end{eg}

The $t_n$-reduction for the non-standard exmKP hierarchy is formulated
below. Since the non-standard case is obtained by setting $N=2$, $r_1=-q_2=1$,
then $\ld_1$ and $\ld_2$ must be zero. Let $\eta:=q-r$, $\psi:=r_x$, then the
$t_n$-reduction of non-standard exmKP hierarchy is
\begin{subequations}
  \label{eqns:n-reduction-non-std-exmKP}
  \begin{align}
    &\mc{L}_{\ta_k}=[\mc{L}^{k/n}_{\ge1}+(q-r)+\pa^{-1}r_x,\mc{L}],\\
    &\mc{L}(q)=0,\\
    &(\pa^{-1}\mc{L}^*)(r_x)=0.
  \end{align}
\end{subequations}
\begin{eg}
  For $n=2$ and $k=3$, $\mc{L}:=\pa^2+V\pa$,
  \begin{align*}
      &V_t=\frac14 V_{xxx}-\frac38 V^2V_x-2(q-r)_x,\\
      &q_{xx}+Vq_{x}=0,\\
      &r_{xx}-Vr_{x}=0.
    \end{align*}
  

  For $n=3$ and $k=2$, $\mc{L}:=\pa^3+V\pa^2+U\pa$,
  \begin{align*}
      &V_{tt}=-\frac13 V_{xxxx}+\frac2{27}(V^3)_{xx}-\frac43
      V_{xx}D^{-1}V_t-\frac43V_xV_t-3(q-r)_{xt}\\
      &\quad\quad-2(V(q-r))_{xx}-6r_{xxx}\\
      &q_{xxx}+Vq_{xx}+Uq_x=0,\\
      &r_{xxx}-(Vr_x)_x+Ur_x=0.
    \end{align*}
  where $U=\frac12(V_x+\frac13 V^2+3\eta+D^{-1}V_t)$.
\end{eg}

\subsection{The $\ta_k$-reduction}
The $\tau_k$-reduction is given by following symmetry constraint
\begin{align*}
  &L^k_{\le0}=\sum q_i\pa^{-1}r_i\pa\\
  &q_{i,t_n}=B_n(q_i)\\
  &r_{i,t_n}=-(\pa^{-1}B_n^*\pa)(r_i)
\end{align*}

It is shown by \cite{XZ05} the constraint is invariant under $t_n$-flow,
therefore we reduce the standard exmKP hierarchy (\ref{eqns:exmKP-Lax}) to the
following well-known \emph{$k$-constrained mKP hierarchy} \cite{OC98,OS93}.
\begin{subequations}
  \begin{align}
    (L_k)_{t_n}&=[(L_k)^{n/k}_{\ge1},L_k],\\
    q_{i,t_n}&=B_n(q_i),\\
    r_{i,t_n}&=-(\pa^{-1}B_n^*\pa)(r_i)
  \end{align}
  where $L_k=L^k_{\ge1}+\sum_iq_i\pa^{-1}r_i\pa$, $B_n=(L_k)^{n/k}_{\ge1}$.
\end{subequations}


In \cite{OS93}, the author discussed the mKP constrained by
$\eta+\pa^{-1}\psi$. Clearly this constraint is just the $\ta_k$-reduction of
non-standard exmKP (\ref{eqns:non-std-exmKP}), by letting $\eta=q-r$,
$\psi=r_x$.


\def\cprime{$'$}

\end{document}